\documentclass{article}
\usepackage{amsmath,amsthm,amssymb,mathrsfs,mathtools,bbm}

\usepackage{hyperref,url}

\usepackage[shortlabels]{enumitem}
\usepackage{color,xcolor,framed,tcolorbox,graphicx}

\definecolor{DarkRed}{rgb}{0.8, 0.0, 0.0} 
\definecolor{DarkGreen}{rgb}{0.0, 0.8, 0.0}

\hypersetup{
    colorlinks   = true,
    linkcolor    = DarkRed,
    citecolor    = DarkGreen
}

\theoremstyle{plain}
    \newtheorem{theorem}{Theorem}
    \newtheorem{proposition}{Proposition} 
    \newtheorem{corollary}[theorem]{Corollary}
    
    \newtheorem{lemma}[theorem]{Lemma}

\theoremstyle{definition}
    \newtheorem{definition}{Definition}
    \newtheorem{example}{\bf Example}
    \newtheorem{remark}{Remark}

\newcommand{\Conv}{\operatorname{Conv}}
\newcommand{\D}{\mathsf{D}}
\renewcommand{\vec}[1]{\boldsymbol{#1}}

\newcommand{\HPhi}{H_{\Phi}}

\newcommand{\DPhi}{D_{\Phi}}

\newcommand{\IPhi}{I_{\Phi}}
\newcommand{\wIPhi}{\widetilde{I}_{\Phi}}

\newcommand{\E}{\mathbb{E}}

\newcommand{\R}{\mathbb{R}}

\newcommand{\fR}{\mathfrak{R}}

\newcommand\independent{\protect\mathpalette{\protect\independenT}{\perp}}
\def\independenT#1#2{\mathrel{\rlap{$#1#2$}\mkern2mu{#1#2}}}
\newcommand{\defeq}{\triangleq}

\makeatletter
\def\rv{\@ifnextchar\bgroup{\rv@i}{\rv@ii}}
\def\rv@i#1{\mathbf{#1}} 
\def\rv@ii#1{\mathbf{#1}} 
\makeatother

\newcommand{\DSBS}{\operatorname{DSBS}}

\newcommand{\diag}{\operatorname{diag}}

\newcommand{\sa}{\mathbb{M}^{\operatorname{sa}}}
\makeatletter
\def\mt{\@ifnextchar\bgroup{\mt@i}{\mt@ii}}
\def\mt@i#1{\boldsymbol{#1}}
\def\mt@ii#1{\boldsymbol{#1}}
\makeatother

\newcommand{\tr}{\operatorname{tr}}
\newcommand{\nortr}{\overline{\operatorname{tr}}}

\newcommand{\cU}{\mathcal{U}}

\newcommand{\cX}{\mathcal{X}}
\newcommand{\cY}{\mathcal{Y}}
\newcommand{\cZ}{\mathcal{Z}}

\usepackage[margin=1in]{geometry}

\newcommand{\set}[1]{\left\{ #1 \right\}}

\newcommand{\PSD}{\mathbb{M}^{+}}

\title{Ribbons from Independence Structure: Hypercontractivity, $\Phi$-Mutual Information, and Matrix $\Phi$-Entropy}

\author{Chenyu Wang, Amin Gohari\\
\emph{\small Department of Information Engineering}\\
\emph{\small The Chinese University of Hong Kong}
}

\begin{document}

\maketitle

\begin{abstract}
    We study the hypercontractivity ribbon and the $\Phi$-ribbon for joint distributions that obey a given independence structure, obtaining tight bounds in some basic regimes. For general independence structures, modeled as a hypergraph whose hyperedges specify mutually independent subcollections of random variables, we provide an explicit inner bound on the $\Phi$-ribbon described by a simple convex hull of incidence vectors. We also provide a new multipartite generalization version and a $\Phi$-mutual information analogue of the Zhang--Yeung inequality, which implies nontrivial points in the hypercontractivity ribbon and the $\Phi$-ribbon respectively. Finally, we propose the matrix $\Phi$-ribbon based on matrix $\Phi$-entropy and establish the tensorization and data processing properties, together with the calculation of an exact matrix SDPI constant for the doubly symmetric binary source. 
\end{abstract}

\section{Introduction}
Given a joint distribution $p_{X^n}$ on $n$ random variables $X^n=(X_1,\dots,X_n)$, the \emph{hypercontractivity (HC) ribbon} is defined as
\begin{equation*}
    \fR(X_1;\dots;X_n)
    \defeq
    \set{
        (\lambda_1,\dots,\lambda_n) \in [0,1]^n \colon
        \sum_{i=1}^n \lambda_i I(U;X_i) \le I(U;X^n), \quad \forall U
    }.
\end{equation*}
The following two facts are well understood. First, for any $p_{X^n}$, one always has
\begin{equation*}
    \set{ (\lambda_1,\dots,\lambda_n)\in[0,1]^n \colon \sum_{i=1}^n \lambda_i \le 1 }
    \subseteq \fR(X_1;\dots;X_n).
\end{equation*}
The above inclusion is tight when $X_1=\cdots=X_n$, in which case the HC ribbon is exactly the simplex $\sum_i\lambda_i\le 1$. At the other extreme, if $X_1,\dots,X_n$ are mutually independent, the HC ribbon equals the full cube $[0,1]^n$. These facts about the two extreme regimes (full dependence and full independence) continue to hold for the $\Phi$-ribbon \cite{beigi2018phi} (which generalizes the HC ribbon). But what about the intermediate regime between the two extremes in which the variables satisfy some partial independence constraint?

This paper is motivated by the following question: \emph{when the independence of $X_1,\dots,X_n$ is only partially constrained, which $(\lambda_1,\dots,\lambda_n)$ are guaranteed to lie in the ribbon for every distribution consistent with that pattern?} Our first contribution provides sharp answers in basic regimes and a general sufficient region for arbitrary independence structures. In particular, if $X_1,\dots,X_n$ are only assumed to be $k$-wise independent (any $k$ of the variables are mutually independent), we show that the guaranteed $\Phi$-ribbon becomes
\begin{align*}
    \Bigl\{(\lambda_1,\dots,\lambda_n)\in[0,1]^n:\ \sum_{i=1}^n \lambda_i\leq k\Bigr\},
\end{align*}
and this description is tight for the HC ribbon. We establish optimality by explicit counterexamples outside this set.

To treat general independence structures, we model the constraints using a hypergraph $H=(V,E)$ on $V=[n]\defeq\{1,\dots,n\}$, where each hyperedge $S\in E$ indicates that $\{X_i:i\in S\}$ are mutually independent. For this setting, we give an explicit inner bound on the ribbon described by the convex hull of the origin, the singleton incidence vectors, and the incidence vectors of hyperedges. 

Next, we provide a connection between the non-Shannon information inequalities and the HC ribbon. In particular, the celebrated Zhang--Yeung inequality \cite{zhang1998characterization} yields nontrivial points in the HC ribbon under a specific case where there are no direct independence structures. A multipartite extension of the Zhang--Yeung inequality was given in \cite{zhang2002new}. We show how such inequalities fit naturally into our framework, provide a streamlined derivation via our general bounds, and present another multipartite extension of the Zhang--Yeung inequality. Moreover, we develop a $\Phi$-mutual information analogue of the Zhang--Yeung inequality, which yields nontrivial $\Phi$-ribbon points. 

Finally, we introduce the \emph{matrix $\Phi$-ribbon} by replacing scalar $\Phi$-entropy with matrix $\Phi$-entropy in the sense of Chen and Tropp \cite{tropp2014subadditivity}. To the best of our knowledge, this matrix $\Phi$-ribbon notion has not appeared previously. We establish the tensorization and the data processing properties for the matrix $\Phi$-ribbon and we compute the exact strong data processing inequality (SDPI) constant for the doubly symmetric binary source (DSBS). Moreover, some of the results from the previous sections can be extended to the matrix $\Phi$-ribbon.

\section{Preliminaries}

\subsection{$\Phi$-ribbon}
The $\Phi$-divergence (usually called $f$-divergence in the literature) was originally proposed by Ali and Silvey \cite{ali1966general} as well as Csiszár \cite{csiszar2011information,csiszar1967information}:
\begin{align*}
    \DPhi(q_X\|p_X)
    &\defeq \sum_{x} p(x) \Phi\left(\frac{q(x)}{p(x)}\right)-\Phi(1).
\end{align*}
Similarly, the $\Phi$-mutual information  can be defined as 
\begin{align*}
    \IPhi(X;Y) \defeq \DPhi(p_{XY}\|p_Xp_Y)=\sum_{x,y} p(x)p(y) \Phi\left(\frac{p(x,y)}{p(x)p(y)}\right) -\Phi(1).
\end{align*}
The concept of $\Phi$-mutual information is related to $\Phi$-entropy, defined as follows:
given a convex function $\Phi$ on some domain $\mathcal{D}$, a function $f:\mathcal{X}\rightarrow \mathcal{D}$, and a random variable $X$ on $\mathcal{X}$,  the $\Phi$-entropy of $f$ is defined to be (refer to \cite{Boucheron2013concentration,chafai2004entropies,chafai2006binomial})
\begin{align*}
    \HPhi(f) &= \E[\Phi(f)] - \Phi(\E[f]).
\end{align*}
Given some arbitrary $p(x,y)$, define $f_x(y)=\frac{p(x,y)}{p(x)p(y)}$. Then  $\IPhi(X;Y)=\sum_{x} p(x) \HPhi(f_x)$ where $\HPhi(f_x)$ is calculated using $p(y)$ for every $x$. Moreover, $\IPhi(X;Y)=\sum_x p(x) \DPhi(p_{Y|x}\|p_Y)$. Thus the $\Phi$-entropy, the $\Phi$-divergence and the $\Phi$-mutual information are related when $\E[f]=1$ and $f\geq 0$. 

Given a function $f_{XY}$ of two random variables $(X,Y)$, we define 
\begin{align}
    H_\Phi(f|Y)& =\E[\Phi(f)]-\mathbb{E}_Y[\Phi(\E[f|Y])]\nonumber\\
    &= \sum_{y}p(y)\big(\E[\Phi(f)|Y=y] - \Phi(\E [f|Y=y])\big).\label{conditionalPhiEnt}
\end{align}
Then we have the chain rule $\HPhi(f)=\HPhi(\E[f|Y])+\HPhi(f|Y)$.

\begin{definition} \label{def:class-F}
    We define $\mathscr{F}$ to be the class of smooth convex functions $\Phi$, whose domain is a convex subset of $\R$, that are not affine (not of the form $at+b$ for some constants $a$ and $b$) and satisfy one of the following equivalent conditions (see \cite[Exercise 14.2]{boucheron2003concentration}):
    \begin{enumerate}[(i)]
    \item $(s,t)\mapsto p\Phi(s)+(1-p)\Phi(t)-\Phi\!\bigl(ps+(1-p)t\bigr)$, for any $p\in[0,1]$, is jointly convex.
    \item $(s,t)\mapsto \Phi(s)-\Phi(t)-\Phi'(t)(s-t)$ is jointly convex.
    \item $(s,t)\mapsto \bigl(\Phi'(s)-\Phi'(t)\bigr)(s-t)$ is jointly convex.
    \item $(s,t)\mapsto \Phi''(s)t^{2}$ is jointly convex.
    \item $1/\Phi''$ is concave.
    \item $\Phi''''\,\Phi'' \ge 2\bigl(\Phi'''\bigr)^{2}$.
    \end{enumerate}
\end{definition}

The class $\mathscr{F}$ implies the subadditivity of the $\Phi$-entropy \cite{beigi2018phi}. Suppose $\Phi\in \mathscr{F}$. If $X_1,\dots,X_n$ are mutually independent, for any $f$, we have
\begin{align*}
    \HPhi(f)\leq \sum_{i=1}^n \HPhi\left(f|X_{\widehat{i}}\right)
\end{align*}
where $X_{\widehat{i}}=(X_1,\dots,X_{i-1},X_{i+1},\dots,X_n)$. 

For $\Phi\in \mathscr{F}$, the $\Phi$-ribbon \cite{beigi2018phi} of random variables $X_1,\dots,X_n$ is defined as the set of all $(\lambda_1,\dots,\lambda_n) \in [0,1]^n$ that 
\begin{align*}  
    \sum_{i=1}^n \lambda_i \HPhi(\E[f|X_i])  
    &\leq \HPhi(f),
\end{align*}
for any function $f$. When we restrict that $f\geq 0$ and $\E[f]=1$, the above region is equivalent to 
\begin{align*}
    \fR_{\Phi}(X_1;\dots;X_n)
    \defeq
    \bigg\{ (\lambda_1,\dots,\lambda_n) \in [0,1]^n:
    \sum_{i=1}^n \lambda_i \IPhi(U;X_i)  
    \leq \IPhi(U;X^n)\quad \forall U
    \bigg\}
\end{align*}
(see \cite[Theorem 14]{beigi2018phi}). In this paper, we will adopt the definition $\fR_{\Phi}(X_1;\dots;X_n)$ as the $\Phi$-ribbon since the use of the auxiliary $U$ aligns with the HC ribbon. We can recover the HC ribbon when $\Phi(t)=t\log t$. 

The following  results for the $\Phi$-ribbon are given in \cite[Example 10, Proposition 11, Theorem 14]{beigi2018phi}: 
\begin{proposition}
\label{prop:phi-ribbon-properties}
    For $\Phi\in \mathscr{F}$, we have 
    \begin{enumerate}[\rm (i)]
    \item $\{(\lambda_1,\dots,\lambda_n):\sum_{i=1}^n \lambda_i \leq 1\}$ is always a subset of the $\Phi$-ribbon for any $p_{X^n}$. Moreover, when $X_1=\dots=X_n$ the above set is exactly the $\Phi$-ribbon.
    \item If $X_1,\dots,X_n$ are mutually independent, the $\Phi$-ribbon is exactly $[0,1]^n$. 
\end{enumerate}
\end{proposition}

\subsection{Matrix $\Phi$-entropy}
Tropp and Chen \cite{tropp2014subadditivity} extend the $\Phi$-entropy to the matrix-valued functions. Consider $\mt F: \cX\to \PSD$ which is a matrix-valued function. For the diagonal matrix $\mt S=\diag(s_1,\dots,s_d)$, define $\Phi(\mt S)\defeq \diag(\Phi(s_1),\dots,\Phi(s_d))$. The positive semidefinite matrix $A$ admits the diagonalization $\mt A =\mt U \mt S\mt U^{\ast}$. Then define $\Phi(\mt A)\defeq \mt U \Phi(\mt S) \mt U^{\ast}$. Let $\sa$ be the set of Hermitian matrices and $\PSD$ be the set of positive semi-definite matrices (PSD). A matrix valued function $\Phi$ is differentiable at $\mt A\in \sa$ if there exists a linear operator $\D\Phi(\mt A)(\cdot)$ such that for all $\mt H\in \sa$, 
    \begin{align*}
        \|\Phi(\mt A+\mt H)-\Phi(\mt A)-\D\Phi(\mt A)(\mt H)\|=o(\|\mt H\|).
    \end{align*}
The linear operator $\D\Phi(\mt A)(\cdot)$ is called the Frech\'et derivative of $\Phi$ at $\mt A$. Let $\D\Phi(A)(\cdot,\cdot)$ denote the second order Frech\'et derivative. 

In \cite{tropp2014subadditivity}, the matrix $\Phi$-entropy of the $d$-dimensional matrix function $\mt F: \cX\to \PSD$ is defined as 
\begin{align*}
    \HPhi(\mt{F})\defeq \nortr(\E[\Phi(\mt{F})]-\Phi(\E[\mt{F}])),
\end{align*}
where $\nortr\defeq \tr/d$ is the normalized trace for $d\times d$ matrices.

It is easy to verify that given $\mt F=\mt F(X)$, we have the chain rule for the $\Phi$-entropy, 
\begin{align} \label{matrix-Phi:chain-rule}
    \HPhi(\mt F)=\HPhi(\E[\mt F|X])+\HPhi(\mt F|X),
\end{align}
where 
\begin{align*}
    \HPhi(\mt F|X)\defeq \sum_{x\in \cX} p(x)\HPhi(\mt F|X=x).
\end{align*}
Since the trace function is convex, the matrix $\Phi$-entropy is non-negative. Suppose $\mt F=\mt F(X,Y,Z)$. By the chain rule \eqref{matrix-Phi:chain-rule}, we have
\begin{align*}
    \HPhi(\mt F|X)=\HPhi(\mt F|XY)+\HPhi(\E[\mt F|XY]|X).
\end{align*}
It implies that 
\begin{align*}
    \HPhi(\mt F|X)\geq \HPhi(\mt F|XY).
\end{align*}
Again by the chain rule \eqref{matrix-Phi:chain-rule}, it is equivalent to 
\begin{align}\label{matrix-Phi:cond-reduce}
    \HPhi(\E[\mt F|XY])\geq \HPhi(\E[\mt F|X]).
\end{align}
Since $\E[\mt F|XY]$ is a mapping of $(X,Y)$ and $\E[\mt F|X]$ is a mapping of $X$, the above inequality means that conditioning reduces the matrix $\Phi$-entropy, aligning with Shannon entropy.

\begin{definition} \label{def:class-FM}
    Define $\mathscr{F}_M$ to be the set of $\Phi:\R_{+}\rightarrow \R$ such that $\Phi$ is continuous and convex, has twice continuous derivatives on $\R_{++}$ that are not affine, and satisfies one of the following equivalent conditions (see \cite{cheng2016characterizations} Theorem 3.3): 
    \setlist{label=\rm (\roman*)}
    \begin{enumerate}
        \item Define $\Psi(t)=\Phi'(t)$. The Frech\'et derivative $\D\Psi$ is invertible and the mapping $A\mapsto (\D\Psi(A))^{-1}$ is operator concave. 
        \item The mapping $(\mt A,\mt B)\mapsto \nortr(t\Phi(\mt A)+(1-t)\Phi(\mt B)-\Phi(t\mt A+(1-t)\mt B))$ is jointly convex for all $0\leq t\leq 1$. \label{def-joint-convexity1}
        \item The matrix Breg\'eman divergence $(\mt A,\mt B)\mapsto \tr(\Phi(\mt A+\mt B)-\Phi(\mt A)-\D\Phi(\mt A)(\mt B))$ is jointly convex. 
        \item $(\mt A,\mt B)\mapsto \tr(\D\Phi(\mt A+\mt B)-\D\Phi(\mt A)(\mt B))$ is jointly convex.
        \item $(\mt A,\mt B)\mapsto \tr(\D^2\Phi(\mt A)(\mt B,\mt B))$ is jointly convex.
    \end{enumerate}
\end{definition}
As is remarked by Cheng and Hsieh \cite{cheng2016characterizations}, the above equivalent conditions can be seen as the non-commutative version of $\mathscr{F}$. As  shown in \cite{tropp2014subadditivity}, for $\Phi\in \mathscr{F}_M$, we have the subadditivity of the matrix $\Phi$-entropy: if $X_1,\dots,X_n$ are mutually independent, for any $\mt F$, 
\begin{align*}
    \HPhi(\mt F)\leq \sum_{i=1}^n \HPhi\left(\mt F|X_{\widehat{i}}\right).
\end{align*}

\begin{example}
    The following functions belong to the $\mathscr{F}_M$ class:
    \begin{enumerate}
        \item $t\log t$. In this case, consider a mapping $\mt F: \cX\to \PSD$ such that $\tr(F(x))=1$ for all $x\in\mathcal{X}$. Then, one can view $\mt F(x)$ as the state of a d-dimensional quantum system $U$, conditioned on $X=x$. So, $(X,U)$ is a cq state. Then, $\E[\mt{F}]$ is the average state of $U$ and the matrix $\Phi$-entropy is $$\HPhi(\mt{F})=\nortr(\E[\mt{F}\log\mt{F}]-\E[\mt{F}]\log(\E[\mt{F}]))=\frac{1}{d}(-H(U|X)+H(U))=\frac{1}{d}I(U;X).$$
        \item $t^p/p$ for $p\in [1,2]$.
    \end{enumerate}
\end{example}

\section{Ribbons under a partial independence structure}

\subsection{Special independence structures}

Suppose $X_1,\dots, X_n$ are $k$-wise independent. We show that the HC ribbon satisfies $\sum_{i=1}^n \lambda_i \leq k$; moreover, it is optimal in the sense of the following theorem. 
\begin{theorem} \label{thm:k-HC}
    Suppose $U,X_1,\dots,X_n$ are random variables. Then 
    \setlist[1]{label=\rm(\roman*)}
    \begin{enumerate}
        \item for any $U,X_1,\dots,X_n$ 
        \begin{align} \label{ineq:general-hc-k}
            \sum_{i=1}^n \lambda_i I(U;X_i)-I(U;X^n)\leq 
            \sum_{S\subseteq [n]:|S|=k} 
            \Bigg(
            \sum_{i\in S} H(X_i)
            -H(X_S)
            \Bigg),
        \end{align}
        for $\lambda_i\in [0,1]$ and $\sum_{i=1}^n \lambda_i \leq k$.
    \end{enumerate}
    Suppose for each subset $S\subseteq [n]$ that $|S|=k$, $(X_i)_{i\in S}$ are mutually independent. Then
    \begin{enumerate}[resume]
        \item if $\sum_{i=1}^n \lambda_i\leq k$, $(\lambda_1,\dots,\lambda_n)$ is always in the hypercontractivity ribbon for any $X_1,\dots,X_n$.
        \item if $\sum_{i=1}^n \lambda_i>k$, there are counterexamples that $(\lambda_1,\dots,\lambda_n)$ is not in the hypercontractivity ribbon.  
    \end{enumerate}
\end{theorem}

\begin{proof}
    (i) It suffices to prove the case $\sum_{i=1}^n \lambda_i=k$. Note that
\[
\Lambda=\{(\lambda_1,\dots,\lambda_n)\in [0,1]^n : \sum_{i=1}^n \lambda_i = k\}
\]
is a polytope in $\mathbb{R}^n$ defined by the constraints $\lambda_i \geq 0$, $\lambda_i \leq 1$ for $i=1,2,\dots,n$, and $\sum_{i=1}^n \lambda_i = k$. The vertices of the polytope must lie on $n$ hyperplanes, meaning that they must satisfy $n$ constraints with equality. Among these, $n-1$ constraints must be selected from $\lambda_i \geq 0$ and $\lambda_i \leq 1$, so $n-1$ entries are integers (either $0$ or $1$). Since $\sum_{i=1}^n \lambda_i = k$ is an integer, all entries must be integers. Thus, the vertices of the polytope correspond to subsets of $[n]$ of size $k$. 

Let $S$ be a subset of $[n]$. Then, $H(X_S|U)\leq \sum_{i\in\mathcal{S}}H(X_i|U)$ where $X_S=(X_i)_{i\in S}$. Equivalently,
\begin{align*}
    \sum_{i\in S} I(U;X_i)
    -I(U;X_S)\leq \sum_{i\in S} H(X_i)-H(X_S)
\end{align*}
From here, we obtain
\begin{align}
    \sum_{i\in S} I(U;X_i)
    -I(U;X^n)\leq \sum_{i\in S} H(X_i)-H(X_S)\label{eqnfe}
\end{align}
As discussed above, we can write a given vector $(\lambda_i)_{i=1}^n\in\Lambda$ as a convex combination of the indicator functions of subsets $S$ of size $k$. Assume that the weight of the subset $S$ in the convex combination is $\omega_S\in[0,1]$. From \eqref{eqnfe}, we obtain
  \begin{align}
        \sum_{i=1}^n \lambda_i I(U;X_i)-I(U;X^n)&\leq 
            \sum_{S\subseteq [n]:|S|=k} \omega_S
            \Bigg(
            \sum_{i\in S} H(X_i)
            -H(X_S)
            \Bigg)\\
            &\leq 
            \sum_{S\subseteq [n]:|S|=k} 
            \Bigg(
            \sum_{i\in S} H(X_i)
            -H(X_S)
            \Bigg)
    \end{align}

    (ii) is a direct consequence of (i).

    (iii) Suppose $\sum_{i=1}^n \lambda_i=k+\epsilon$ for $\epsilon>0$. Let $p>n$ be prime. Let $A_0,\dots,A_{k-1}$ be mutually independent random variables uniformly distributed from $\{0,\dots,p-1\}$. Define $P(x)\defeq A_0+A_1 x+\dots+A_{k-1} x^{k-1}\mod p$ and let $X_i=P(i)$. Then any $k$ of $X_1,\dots,X_n$ are mutually independent. Let $U=A_0\dots A_{k-1}$. Then
    \begin{align*}
        \sum_{i=1}^n \lambda_i I(U;X_i)-I(U;X^n)
        &= \sum_{i=1}^n \lambda_i H(X_i) - H(A_0\dots A_{k-1})\\
        &= \log(p)\sum_{i=1}^n \lambda_i  - k\log(p)\\
        &= \epsilon \log(p)\\
        &>0.
    \end{align*}
    Thus $(\lambda_1,\dots,\lambda_n)$ is not in the hypercontractivity ribbon. 
\end{proof}

With the same polytope claim and combined with Proposition \ref{prop:phi-ribbon-properties} (ii), we can generalize the above results about the achievable region to the $\Phi$-ribbon:
\begin{theorem} \label{thm:k-ind}
    Suppose $\Phi\in \mathscr{F}$. Let  $X_1,\dots,X_n$ be $k$-wise independent. Then
    $\{(\lambda_1,\dots,\lambda_n)\in [0,1]^n:\sum_{i=1}^n \lambda_i\leq k\}\subseteq \fR_{\Phi}(X_1;\dots;X_n)$ for any $X_1,\dots,X_n$.  
\end{theorem}
\begin{proof}
    For any $S\subseteq [n]$ that $|S|=k$, $X_S=(X_i)_{i\in S}$ are mutually independent. By Proposition \ref{prop:phi-ribbon-properties} (ii), 
    \begin{align}
        \sum_{i\in S} \IPhi(U;X_i)
        \leq \IPhi(U;X_S) \leq \IPhi(U;X^n), \qquad\forall S\subseteq [n]: |S|=k.\label{eqncv2}
    \end{align}
As discussed in the proof of Theorem \ref{thm:k-HC}, the vertices of the set    \[
    \Lambda=\{(\lambda_1,\dots,\lambda_n)\in [0,1]^n : \sum_{i=1}^n \lambda_i = k\}
    \]
    correspond to subsets of $[n]$ of size $k$.  Since the inequality~\eqref{eqncv2} holds, by taking the convex combination of the above inequalities for all $S \subseteq [n]$ with $|S| = k$ (the vertices of the polytope), we obtain
    \begin{align*}
        \sum_{i=1}^n \lambda_i \IPhi(U;X_i)
        \leq \IPhi(U;X^n),
    \end{align*}
    for any $(\lambda_1,\dots,\lambda_n)$ in the polytope.
\end{proof}

\subsection{General independence structures}
Given a hypergraph $H=(V,E)$, we say that $H$ specifies the independence structure of the random variables $(X_1,\dots,X_n)$ if for $S\in E$, each $X_i$ for $i\in S$ are mutually independent. 

\begin{theorem} \label{thm:general-hypergraph}
    Suppose $\Phi\in \mathscr{F}$. Let $H=(V,E)$ be the hypergraph specifies the independence structure of $X_1,\dots,X_n$. Define $$F\defeq\{\vec{0},\vec{e}_1,\dots,\vec{e}_n\}\cup \{\vec{e}_{S}\}_{S\in E}$$
    where $\vec{e}_i\in\mathbb{R}^n$ stands for the $i$-th standard basis and $\vec{e}_S=\sum_{i\in S}\vec{e}_i$. Let $\Conv(F)$ denote the convex hull of $F$. Then 
    $$\Conv(F)\subseteq \fR_{\Phi}(X_1;\dots;X_n).$$ 
    
\end{theorem}
\begin{proof}
    For each $i\in [n]$, we have
    \begin{align} \label{ineq:proof-i}
        \IPhi(U;X_i) \leq\IPhi(U;X^n).
    \end{align}
    By Proposition \ref{prop:phi-ribbon-properties} (ii), for each $S\in E$, we have 
    \begin{align} \label{ineq:proof-S}
        \sum_{i\in S} \IPhi(U;X_i) \leq \IPhi(U;X_S)\leq \IPhi(U;X^n).
    \end{align}
    Taking the convex combination of \eqref{ineq:proof-i} and  \eqref{ineq:proof-S} for each $i$ and $S$, we obtain the desired inequality. 
\end{proof}

\section{Hypercontractivity and the Zhang--Yeung inequality}

The Zhang--Yeung inequality states that for any four random variables $Z$, $U$, $X$, and $Y$ (\cite{zhang1998characterization}, \cite[Theorem 14.7]{yeung2012first}):
\begin{align*}
    I(Z; U) \leq &I(Z; U|X)+I(Z; U|Y )+I(X; Y)+I(X; Z|U)+I(X; U|Z)+I(Z; U|X)
\end{align*}
and 
\begin{align*}
    2I(X;Y) \leq &I(Z;U) + I(Z;XY) + 3I(X;Y|Z) + I(X;Y|U).
\end{align*}
Moreover, a generalization of the Zhang-Yeung inequality to a 5 random variables ineuality was given in \cite[Theorem 2]{makarychev2002new}: 
\begin{align*}
    I(Z;U) \leq I(Z;U|X)+I(Z; U|Y)+I(X;Y)+I(Z;V|U)+I(Z; U|V)+I(U;V|Z).
\end{align*}

Renaming the variables $X\to S$,
$Y\to T$,
$Z\to X$,
$U\to Y$,
$V\to J$, the 5-variable non-Shannon type inequality can be rewritten in the following form: 
\begin{align*}
    \frac{2}{3}I(J;X)
    +\frac{2}{3}I(J;Y)-I(J;XY)\leq 
    \frac{1}{3}
    \bigg(I(S;T)+I(X;Y|S) +I(X;Y|T)\bigg).
\end{align*}
We write $A\independent B|C$ to denote that $A$ and $B$ are conditionally independent given $C$. Zhang and Yang \cite{zhang2002new} define $(S,T)$ to be the independent perfect coverings of $(X,Y)$ if $S\independent T$, $X\independent Y|S$ and $X\independent Y|T$. They also provide a non-trivial example of $(X,Y)$ such that there exist $(S,T)$ satisfying the above conditions. The existence of the independent perfect coverings implies that  $(2/3,2/3)$ is  in the HC ribbon of $(X,Y)$. 

\subsection{Multipartite generalization}
Zhang and Yang also provide a multipartite generalization of the Zhang-Yeung inequality \cite[Theorem 3]{zhang2002new}: 
\begin{align*}
    n I(J;X;Y)
    -I(J;XY)
    \leq \sum_{i=1}^n I(X;Y|T_i)
    +\sum_{i=1}^n H(T_i)-H(T^n),
\end{align*}
where $I(J;X;Y)=I(J;X)-I(J;X|Y)$. We will give another generalized version. 

\begin{theorem}
    For any random variables $U,X_1,\dots,X_n,S,T$, we have
    \begin{align*}
        \sum_{i=1}^n \gamma_i I(U;X_i)
        -I(U;X^n)
        \leq 
        \frac{1}{3}
        \Bigg(
        \sum_{i<j}
        \bigg(
        I(X_i;X_j|S)
        +I(X_i;X_j|T)
        \bigg)
        +I(S;T)
        \Bigg)
    \end{align*}
    if $(\gamma_1,\dots,\gamma_n)\in [0,1]^n$ and $\sum_{i=1}^{n} \gamma_i\leq \frac{4}{3}$. 
\end{theorem}
\begin{remark}
    If there exist independent $S,T$ such that given $(S,T)$, $(X_1,\dots,X_n)$ are pairwise independent, then $\{(\gamma_1,\dots,\gamma_n)\in [0,1]^n:\sum_{i=1}^n \gamma_i\leq \frac{4}{3}\}$ is a subset of the HC ribbon for $X_1,\dots,X_n$. 
\end{remark}
\begin{proof} Since the inequality depends only on the marginal distribution of $(U,X^n)$ and $(S,T,X^n)$, without loss of generality we may assume the Markov chain
\[U\rightarrow X^n\rightarrow (S,T).\]
    Applying Theorem \ref{thm:k-HC}, we have 
    \begin{align*}
        \sum_{i=1}^n \lambda_i I(U;X_i|S)
        -I(U;X^n|S)
        \leq \sum_{i<j} I(X_i;X_j|S)
    \end{align*}
    for $\lambda_i\geq 0$, $\sum_{i=1}^n \lambda_i=2$. By the chain rule, 
    \begin{align*}
        \sum_{i=1}^n \lambda_i 
        \bigg(
        I(U;X_i)+I(U;S|X_i)-I(U;S)
        \bigg)
        -I(U;X^n|S)
        \leq \sum_{i<j} I(X_i;X_j|S).
    \end{align*}
    By the nonnegativity of the mutual information, we have 
    \begin{align} \label{ineq:zy-1}
        \sum_{i=1}^n \lambda_i 
        I(U;X_i)
        -I(U;S)
        -I(U;X^nST)
        \leq \sum_{i<j} I(X_i;X_j|S).
    \end{align}
    Similarly, we have
    \begin{align} \label{ineq:zy-2}
        \sum_{i=1}^n \lambda_i 
        I(U;X_i)
        -I(U;T)
        -I(U;X^nST)
        \leq \sum_{i<j} I(X_i;X_j|T).
    \end{align}
    Moreover, we have
    \begin{align} \label{ineq:zy-3}
        I(U;S)+I(U;T)-I(U;X^nST)\leq I(S;T).
    \end{align}
    Summing up \eqref{ineq:zy-1}, \eqref{ineq:zy-2} and \eqref{ineq:zy-3}, we obtain
    \begin{align*}
        \sum_{i=1}^n 2\lambda_i I(U;X_i)
        -3I(U;X^nST)
        \leq 
        \sum_{i<j}
        \bigg(
        I(X_i;X_j|S)
        +I(X_i;X_j|T)
        \bigg)
        +I(S;T).
    \end{align*}
    Dividing both sides by $3$ and applying the Markov chain $U\to X^n \to ST$, we obtain the desired result. 
\end{proof}

\subsection{Generalization to the $\Phi$-mutual information}

We define the conditional $\Phi$-mutual information as
\begin{align*}
    \IPhi(X;Y|Z) \defeq \IPhi(X;YZ)-\IPhi(X;Z)\geq 0.
\end{align*}
While $\IPhi(X;Y)=\IPhi(Y;X)$,  the conditional $\Phi$-mutual information is not symmetric, i.e, $\IPhi(X;Y|Z)$ may not be equal to $\IPhi(Y;X|Z)$ in general. To the best of our knowledge, this definition of the conditional $\Phi$-mutual information is new. In general, $\IPhi(X;Y|Z)$ is not equal to $\sum_{z}p(z) \IPhi(X;Y|Z=z)$. However, $\IPhi(X;Y|Z)$ can be expressed in terms of the conditional $\Phi$-entropy (which has an averaging form)  as defined in \eqref{conditionalPhiEnt}: for $x\in \cX$, let $f_x(y,z)=\frac{p(x,y,z)}{p(x)p(y,z)}$. Then, one can verify that
\begin{align*}
\IPhi(X;Y|Z)&=\sum_{x} p(x)
    \HPhi(f_x(Y,Z)|Z)=\sum_{x} p(x)
    \sum_{z}p(z)\HPhi(f_x(Y,z)|z).
\end{align*}

\begin{lemma}\label{lmm5ni}
    We have the following properties of the $\Phi$-mutual information: 
    \setlist[1]{label=\rm(\roman*)}
\begin{enumerate}[\rm (i)]
    \item If $Y$ is a function of $X$, i.e., $Y=g(X)$,  $\IPhi(X;Y)=\IPhi(Y;Y)$. In particular, $\IPhi(XY;Y)=\IPhi(Y;Y)$.
    \item $\IPhi(X;Y|X)=0$.
    \item $\IPhi(X;XY|Z)=\IPhi(X;X|Z)$.
    \item $\IPhi(X;YY)=\IPhi(X;Y)$, \item $\IPhi(X;Y)\leq \min(\IPhi(X;X),\IPhi(Y;Y))$.
\end{enumerate}
For any $\Phi\in \mathscr{F}$, we have the following additional property:
\begin{enumerate}[resume]
 \item If random variables $X$ and $Y$ are independent, then $\IPhi(U;Y|X)\geq \IPhi(U;Y)$ for all $P_{U|XY}$. More generally, if $X\to Z\to Y$ forms a Markov chain, then
\begin{align} 
    \IPhi(U;Y|XZ)\geq \IPhi(U;Y|Z).
\end{align}
\end{enumerate}
\end{lemma}
The proof is in Appendix \ref{app:proof-lmm5ni}.

We define the \emph{max-$\Phi$-mutual information} as 
\begin{align*}
    \wIPhi(X;Y) &\defeq \max_{U} \IPhi(Y;U)-\IPhi(U;Y|X) \\
    &= \max_{U} \IPhi(U;X)+\IPhi(U;Y)-\IPhi(U;XY).
\end{align*}
In Appendix \ref{app:proof-optimizer}, we show that to evaluate the above maximum, it suffices to consider random variables $U$ with cardinality at most $|\mathcal{X}||\mathcal{Y}|$. Then we can define the \emph{conditional max-$\Phi$-mutual information} : 
\begin{align} \label{opt:cond-IPhi}
    \wIPhi(X;Y|Z) 
    &= \max_{U} \IPhi(U;X|Z) +\IPhi(U;Y|Z) -\IPhi(U;XY|Z)\nonumber \\
    &= \max_{U} \IPhi(U;Y|Z)-\IPhi(U;Y|XZ) \nonumber\\
    &= \max_{U} \IPhi(U;XZ)+\IPhi(U;YZ)-\IPhi(U;XYZ)-\IPhi(U;Z).
\end{align}
Similarly, he cardinality of $U$ is at most $|\cX||\cY||\cZ|$. 

Although the optimizer $U$ for $\wIPhi(X;Y|Z)$ is not known for the general convex $\Phi$. We provide a class of $\Phi$ such that the optimizer $U$ can be restricted. In particular, to evaluate $\wIPhi(X;Y)$ the optimizer $U$ equals $XY$.
\begin{definition} \label{def:class-G}
    We define $\mathscr{G}$ to be the class of smooth convex functions $\Phi$, whose domain is a convex subset of $\R_+$, that are not affine and satisfy that $s^2\Phi''(s)$ is concave.
\end{definition}
Examples of functions in class $\mathscr{G}$ include $x\log(x), -\log(x), (\sqrt{x}-1)^2, x \log x - (1+x)\log(\frac{1+x}{2})$ and $\frac{1}{\alpha(\alpha-1)}(x^\alpha - 1)$ for $\alpha\in(0,1)$.
\begin{theorem} \label{thm:wIPhi-optimizer}
    Suppose $\Phi\in \mathscr{G}$. To solve \eqref{opt:cond-IPhi}, it suffices to consider $p_{U|XYZ}$ such that $H(XY|UZ)=0$.
\end{theorem}
The proof is in Appendix \ref{app:proof-optimizer}. 
\begin{corollary}
    Suppose $\Phi\in \mathscr{G}$. Then
    \begin{align*}
        \wIPhi(X;Y) =\IPhi(X;X)+\IPhi(Y;Y)-\IPhi(XY;XY). 
    \end{align*}
\end{corollary}
For $\Phi(t)=t\log t$, the right hand side reduces to $I(X;Y)$. 
\begin{remark}
    In a concurrent work \cite{AbinSubadditiveInformation}, the second author and his collaborators show that for every $\Phi\in\mathscr{G}$ and for any probability mass functions $P_X,Q_X$ on $\cX$ and $P_Y,Q_Y$ on $\cY$, we have
    \begin{align*}
        \DPhi\big(P_X\otimes P_Y \,\big\|\, Q_X\otimes Q_Y\big)
        \le \DPhi(P_X\|Q_X) + \DPhi(P_Y\|Q_Y).
    \end{align*}
\end{remark}

Now we show that $\wIPhi(X;Y)$ and $\wIPhi(X;Z|Y)$ are true independence measures. 
\begin{theorem} \label{thm:tilde-IPhi}
    We have the following properties for any convex $\Phi$: 
    \setlist[1]{label=\rm(\roman*)}
    \begin{enumerate}
        \item For any $X$, $Y$ and $Z$, $\wIPhi(X;Y|Z)\geq \IPhi(X;Y|Z)\geq 0$, 
        \item If $\Phi$ is strictly convex at $x=1$, $\wIPhi(X;Y)=0$ implies independence of $X$ and $Y$,
        \item $\wIPhi(X;Y)\leq\wIPhi(X;X)=\IPhi(X;X)$ and $\wIPhi(X;XY)=\wIPhi(X;X)=\IPhi(X;X)$.
    \end{enumerate}
    For any $\Phi\in \mathscr{F}$, we have the following additional properties: 
    \begin{enumerate}[resume]
        \item If random variables $X$ and $Y$ are independent, then $\wIPhi(X;Y)=0$. 
        \item If random variables $X$, $Y$ and $Z$ form a Markov chain $X\to Z\to Y$, then $\wIPhi(X;Y|Z)=0$. 
    \end{enumerate}
\end{theorem}
The proof is in Appendix \ref{app:proof-thm-IPhi}.

Then we can derive the $\Phi$-mutual information generalization of the Zhang--Yeung inequality. 
\begin{theorem}
    Suppose $\Phi$ is convex. For any random variables $J,X,Y,S,T$, we have
    \begin{align*}
        \frac{2}{3} \IPhi(J;X)
        +\frac{2}{3} \IPhi(J;Y)
        -\IPhi(J;XY)
        \leq \frac{1}{3} \bigg(
        \wIPhi(S;T)
        +\wIPhi(X;Y|S)
        +\wIPhi(X;Y|T)
        \bigg).
    \end{align*}
    Suppose $\Phi\in \mathscr{F}$. If there exists $S,T$ such that $S\independent T$, $X\independent Y|S$ and $X\independent Y|T$, then $(2/3,2/3)$ is in the $\Phi$-ribbon. 
\end{theorem}
\begin{proof} Since the expression depends only on the marginal distribution on $(J,X,Y)$ and $(X,Y,S,T)$, without loss of generality assume that
$$J\rightarrow (X,Y)\rightarrow (S,T).$$
    By definition, we have
    \begin{align*} 
        \IPhi(J;S)+\IPhi(J;T)
        -\IPhi(J;ST)\leq \wIPhi(S;T),
    \end{align*}
    Then 
    \begin{align} \label{ineq:IPhi-JST}
        \IPhi(J;S)+\IPhi(J;T)
        -\IPhi(J;XYST)\leq \wIPhi(S;T),
    \end{align}
    Again by the definition, 
    \begin{align*} 
        \IPhi(J;X|S)+\IPhi(J;Y|S)
        -\IPhi(J;XY|S)
        \leq \wIPhi(X;Y|S).
    \end{align*}
    Then we obtain
    \begin{align} \label{ineq:IPhi-JXYS}
         \IPhi(J;X)+\IPhi(J;Y) 
         -\IPhi(J;XYTS)-\IPhi(J;S)
         &\leq \IPhi(J;XS)+\IPhi(J;YS) 
         -\IPhi(J;XYS)-\IPhi(J;S)\nonumber \\
         &=\IPhi(J;X|S) +\IPhi(J;Y|S)
        -\IPhi(J;XY|S)\nonumber \\
        &\leq \wIPhi(X;Y|S)
    \end{align}
    where the first line follows from the nonnegativity of $\Phi$-mutual information and the second line follows from the chain rule. 

    Similarly, we have
    \begin{align}\label{ineq:IPhi-JXYT}
        \IPhi(J;X)+\IPhi(J;Y) 
         -\IPhi(J;XYTS)-\IPhi(J;T)
         &\leq \wIPhi(X;Y|T)
    \end{align}

    Summing up \eqref{ineq:IPhi-JST}, \eqref{ineq:IPhi-JXYS} and \eqref{ineq:IPhi-JXYT}, we obtain 
    \begin{align*}
        \frac{2}{3} \IPhi(J;X)
        +\frac{2}{3} \IPhi(J;Y)
        -\IPhi(J;XYTS)
        \leq \frac{1}{3} \bigg(
        \wIPhi(S;T)
        +\wIPhi(X;Y|S)
        +\wIPhi(X;Y|T)
        \bigg). 
    \end{align*}
   Using the Markov chain $J\to XY\to TS$, we obtain
    \begin{align*}
        \frac{2}{3} \IPhi(J;X)
        +\frac{2}{3} \IPhi(J;Y)
        -\IPhi(J;XY)
        \leq \frac{1}{3} \bigg(
        \wIPhi(S;T)
        +\wIPhi(X;Y|S)
        +\wIPhi(X;Y|T)
        \bigg).
    \end{align*}

    Now suppose $S\independent T$, $X\independent Y|S$ and $X\independent Y|T$. By Theorem \ref{thm:tilde-IPhi}, $\wIPhi(S;T)+\wIPhi(X;Y|S)+\wIPhi(X;Y|T)=0$. Thus $(2/3,2/3)$ is in the $\Phi$-ribbon. 
\end{proof}

\section{Matrix $\Phi$-ribbon}

We define the matrix $\Phi$-ribbon for the first time as follows:
\begin{definition}
    For $\Phi\in \mathscr{F}_M$, the matrix $\Phi$-ribbon $\fR_{\Phi}^M(X;Y)$ for the random variables $(X,Y)$ is the set of all $(\lambda_1,\lambda_2)\in [0,1]^2$ such that 
    \begin{align*}
        \lambda_1 \HPhi(\E[\mt{F}|X]) +\lambda_2 \HPhi(\E[\mt{F}|Y]) \leq \HPhi(\mt{F}),
    \end{align*}
    for all PSD matrix-valued $\mt F: \cX\times \cY\to \PSD$ with any dimensions.  
\end{definition}
When we restrict the dimension of $\mt{F}(x,y)$ to be one, Definition \ref{def:class-FM} aligns with Definition \ref{def:class-F}; thus the matrix $\Phi$-ribbon reduces to the ordinary $\Phi$-ribbon. If we assume that $\Phi(t)=t\log t$ and $\mt{F}(x,y)$ has trace 1 for all $x,y$, $\HPhi(\mt F)$, $\HPhi(\E[\mt F|X])$ and $\HPhi(\E[\mt F|Y])$ become the quantum entropy of classical-quantum ensembles (see \cite{cheng2019matrix}). 

The following lemma will be useful when we prove the tensorization and the data processing of the matrix $\Phi$-ribbon. The proofs follow similar lines as in \cite{beigi2018phi}.
\begin{lemma} \label{lem:matrix-Phi-ineq}
    \setlist{label=\rm (\roman*)}
    \begin{enumerate}
        \item \label{matrix-Phi-ineq-1} Assume $X$ and $Y$ are independent random variables and $\mt F=\mt F(X,Y)$ is an arbitrary random positive semidefinite matrix of $(X,Y)$. Then for any $\Phi\in \mathscr{F}_M$, we have
        \begin{align*}
            \HPhi(\mt F|X)\geq \HPhi(\E[\mt F|Y]).
        \end{align*}
        \item \label{matrix-Phi-ineq-2} More generally, if $X,Y,Z$ are three random variables satisfying the Markov chain condition $X\rightarrow Z \rightarrow Y$ and $\mt F=\mt F(X,Y,Z)$, we have
        \begin{align*}
             \HPhi(\mt F|XZ)\geq \HPhi(\E[\mt F|YZ]|Z).
        \end{align*}
    \end{enumerate}
\end{lemma}

\begin{proposition} \label{prop:ribbon-indep}
    For $\Phi\in \mathscr{F}_M$, if $X$ is independent of $Y$, then $\fR_{\Phi}^M(X;Y)=[0,1]^2$. 
\end{proposition}

\begin{theorem} \label{thm:tensor-dpi}
    For any $\Phi\in \mathscr{F}_M$, the matrix $\Phi$-ribbon satisfies the tensorization and the data processing as follows:
    \begin{enumerate}[\rm (i)]
        \item {\rm (Tensorization)} If $(X_1,Y_1)$ is independent of $(X_2,Y_2)$, then
        \begin{align*}
            \fR_{\Phi}^M(X_1X_2;Y_1Y_2) =\fR_{\Phi}^M(X_1;Y_1)\cap \fR_{\Phi}^M(X_2;Y_2).
        \end{align*}
        \item {\rm (Data processing)} If $(X_2,Y_2)$ is generated from $(X_1,Y_1)$ by $p(x_2,y_2|x_1,y_1)=p(x_2|x_1)p(y_2|y_1)$, then 
        \begin{align*}
            \fR_{\Phi}^M(X_1;Y_1) \subseteq \fR_{\Phi}^M(X_2;Y_2).
        \end{align*}
    \end{enumerate} 
\end{theorem}

The proof of the above properties of the matrix $\Phi$-ribbon can be found in Appendix \ref{app:proof-matrix-phi}. 

Let $\eta_{\Phi}^M(X;Y)$ denote the infimum of all constants $c$ such that 
\begin{align*}
    c\HPhi(\mt F) \geq \HPhi(\E[\mt F|Y]),
\end{align*}
for all positive semidefinite $\mt F(X)$. 

Suppose $X,Y$ is a doubly symmetric binary source (DSBS) with parameter $\lambda$, denoted as $\DSBS(\lambda)$ where the joint distribution is given by 
\begin{align*}
    p_{XY}=\begin{pmatrix}
        \frac{1+\lambda}{4}, &\frac{1-\lambda}{4}\\
        \frac{1-\lambda}{4}, &\frac{1+\lambda}{4}
    \end{pmatrix}.
\end{align*}

\begin{proposition} \label{prop:M-Phi-SDPI}
    If $(X,Y)$ are distributed according to $\DSBS(\lambda)$, then for any $\Phi \in \mathscr{F}^M$, we have
    \begin{align*}
        \eta_{\Phi}^M(X;Y)=\lambda^2. 
    \end{align*}
\end{proposition}
The proof is in Appendix \ref{app:proof-SDPI}.

The previous results about the $\Phi$-ribbon can be easily extended to the matrix $\Phi$-ribbon. We state the theorems below without proof. 
\begin{theorem}
    Suppose $\Phi\in \mathscr{F}_M$. Let  $X_1,\dots,X_n$ be $k$-wise independent. Then
    $\{ (\lambda_1,\dots,\lambda_n)\in [0,1]^n:\sum_{i=1}^n \lambda_i\leq k\}\subseteq \fR_{\Phi}^M(X_1;\dots;X_n)$ for any $X_1,\dots,X_n$.
\end{theorem}

\begin{theorem}
    Suppose $\Phi\in \mathscr{F}_M$. Let $H=(V,E)$ be the hypergraph specifies the independence structure of $X_1,\dots,X_n$. Define $$F\defeq\{\vec{0},\vec{e}_1,\dots,\vec{e}_n\}\cup \{\vec{e}_{S}\}_{S\in E}$$
    where $\vec{e}_i$ stands for the $i$-th standard basis and $\vec{e}_S=\sum_{i\in S}\vec{e}_i$. Let $\Conv(F)$ denote the convex hull of $F$. Then 
    $$\Conv(F)\subseteq \fR_{\Phi}^M(X_1;\dots;X_n).$$ 
\end{theorem}

\begin{theorem}
    Suppose $\Phi\in \mathscr{F}_M$. If there exists $S,T$ such that $S\independent T$, $X\independent Y|S$ and $X\independent Y|T$, then $(2/3,2/3)$ is in the matrix $\Phi$-ribbon. 
\end{theorem}

\bibliographystyle{IEEEtran}
\bibliography{refs}

\appendix

\section{Proof of Lemma \ref{lmm5ni}} \label{app:proof-lmm5ni}
Property (i) follows from 
\begin{align*}
    \IPhi(X;Y) &=\sum_{x,y} p(x)p(y)\Phi\left(\frac{p(x,y)}{p(x)p(y)}\right)-\Phi(1)\\
    &=\sum_{x} p(x)
    \Bigg(
    p(g(x))\Phi\left(\frac{1}{p(g(x))}\right)
    +(1-p(g(x))) \Phi(0)
    \Bigg)
    -\Phi(1)\\
    &=\IPhi(Y;Y).
\end{align*}
Property (ii) follows from $\IPhi(X;Y|X)=\IPhi(X;XY)-\IPhi(X;X)=0$. 

Property (iii) follows from
\begin{align*}
    \IPhi(X;XY|Z)&=\IPhi(X;XYZ)-\IPhi(X;Z)\\
    &=\IPhi(X;XZ)-\IPhi(X;Z)\\
    &=\IPhi(X;X|Z).
\end{align*}
Property (iv) follows from
\begin{align*}
    \IPhi(X;YY)&=\sum_{x,y,y'} p(x)p(y,y')\Phi\left(\frac{p(x,y,y')}{p(x)p(y,y')}\right)-\Phi(1)\\
    &=\sum_{x,y} p(x)p(y)\Phi\left(\frac{p(x,y)}{p(x)p(y)}\right)-\Phi(1)\\
    &=\IPhi(X;Y).
\end{align*}

Property (v) follows from
$$\IPhi(X;Y)\leq \IPhi(XY;Y)=\IPhi(Y;Y)$$
where we used the first part of this lemma.

Part (vi). For $u\in \cU$, let $f_u(x,y,z)=\frac{p(u,x,y,z)}{p(u)p(x,y,z)}$. Then, one can verify that
\begin{align*}
\IPhi(U;Y|XZ)&=\sum_{u} p(u)
\HPhi(\E[f_u|XYZ]|XZ)
\\
\IPhi(U;Y|Z)&=\sum_{u} p(u)
    \HPhi(\E[f_u|YZ]|Z)
\end{align*}
It suffices to show that for every $u$, we have
$$\HPhi(\E[f_u|XYZ]|XZ)\geq \HPhi(\E[f_u|YZ]|Z).$$
This follows from \cite[Lemma 7(b)]{beigi2018phi} for the function $g(x,y,z)=\E[f_u|X=x,Y=y,Z=z]$.

\section{Proof of Theorem \ref{thm:wIPhi-optimizer}} \label{app:proof-optimizer}

Consider
\begin{align*}
    \wIPhi(X;Y|Z)
    &= \max_{U} \IPhi(U;XZ)+\IPhi(U;YZ)-\IPhi(U;XYZ)-\IPhi(U;Z).
\end{align*}
Note that
\begin{align*}
    \IPhi(U;XZ)+\IPhi(U;YZ)-\IPhi(U;XYZ)-\IPhi(U;Z)
    &=\sum_{u} p(u) g(p_{XYZ|u})
\end{align*}
where 
\begin{align*}
    g(q_{XYZ})=\DPhi(q_{XZ}\|p_{XZ})
    +\DPhi(q_{YZ}\|p_{YZ}) 
    -\DPhi(q_{XYZ}\|p_{XYZ}) 
    -\DPhi(q_{Z}\|p_{Z}).
\end{align*}
Then $\wIPhi(X;Y|Z)$ is the upper concave envelope of $g(q_{XYZ})$ at $p_{XYZ}$. By Carath\'eodory's theorem, it suffices to consider the random variable $U$ with cardinality at most $|\cX||\cY||\cZ|$. We claim the following statements:
\begin{itemize}
    \item Fixing $q_Z$ and $q_{Y|XZ}$, $g(q_{XYZ})$ is convex in $q_{X|Z}$.
    \item Fixing $q_Z$ and $q_{X|YZ}$, $g(q_{XYZ})$ is convex in $q_{Y|Z}$.
\end{itemize}
To compute 
the upper concave envelope of 
$g(q_{XYZ})$, we need to maximize
\begin{align}\sum_{i}\omega_i g(q^{(i)}_{XYZ})\label{objf}\end{align}
subject to
$$\sum_{i}\omega_i q^{(i)}_{XYZ}=p_{XYZ}.$$
Take some weights $\omega_i$ and some distributions $q^{(i)}_{XYZ}$. 
The first bullet shows that if we write 
$q^{(i)}_{XYZ}$ itself as a convex combination of distributions where $X$ is a function of $Z$, and replace $q^{(i)}_{XYZ}$ with that convex combination, the objective function in \eqref{objf} will not decrease. Thus, to compute the upper concave envelope of $g(q_{XYZ})$, it suffices to take a convex combination of distributions where $X$ is a function of $Z$. Similarly, the second bullet shows that we can then assume that $Y$ is also a function of $Z$. Thus, it suffices to take a convex combination of distributions where  $(X,Y)$ is a function of $Z$ under $q_{XY|Z}$. This would complete the proof. 

It remains to prove the first bullet (proof of the second bullet is similar).

Fixing $q_Z$ and $q_{Y|XZ}$, the term $\DPhi(q_{Z}\|p_{Z})$is fixed; the term $\DPhi(q_{YZ}\|p_{YZ})$ is convex in $q_{X|Z}$ since $\Phi$-divergence is convex in its arguments. For every $z$ and fora fixed $q(y|x,z)$ and $q(z)$,  the mapping
\begin{align*}
    q(x|z) \mapsto 
    \Phi\left(\frac{q(x|z)q(z)}{p(x,z)}\right)
    -\Phi\left(\frac{q(x,y|z)q(z)}{p(x,y,z)}\right)
\end{align*}
is convex 
as the second order condition satisfies
\begin{align*}
    q(x|z)^2 \Phi''\left(\frac{q(x|z)q(z)}{p(x,z)}\right)
    -q(x,y|z)^2 \Phi''\left(\frac{q(x,y|z)q(z)}{p(x,y,z)}\right)
    \geq 0
\end{align*}
since $s^2\Phi''(s)$ is concave. This implies that
$\DPhi(q_{XZ}\|p_{XZ})
    -\DPhi(q_{XYZ}\|p_{XYZ})$ is convex in $q_{X|Z}$.

\section{Proof of Theorem \ref{thm:tilde-IPhi}} \label{app:proof-thm-IPhi}

(i) Taking the special choice of $U=X$ yields
\begin{align*}
    \IPhi(X;X|Z) +\IPhi(X;Y|Z) -\IPhi(X;XY|Z)
    &=\IPhi(X;Y|Z),
\end{align*}
where we used Lemma \ref{lmm5ni}. 
Thus $\wIPhi(X;Y|Z)\geq \IPhi(X;Y|Z)$. 

(ii) Suppose $\wIPhi(X;Y)=0$. By (i), we have
\begin{align*}
    \IPhi(X;Y)\leq \wIPhi(X;Y)=0.
\end{align*}
Then $\IPhi(X;Y)=0$. Since $\Phi$ is strictly convex at $x=1$, $X$ is independent of $Y$ (see \cite[Proposition 1]{sason2016f}).  

(iii) For any $U,X,Y$, we have
\begin{align*}
    \IPhi(U;X)+\IPhi(U;Y)-\IPhi(U;XY)&\leq\IPhi(U;X)\\
    &=\IPhi(U;X)+\IPhi(U;X)-\IPhi(U;XX),
\end{align*}
where we used Lemma \ref{lmm5ni}.
Thus, $\wIPhi(X;Y)\leq \wIPhi(X;X)$. Next,
\begin{align*}
\wIPhi(X;X)&=\max_U\IPhi(U;X)+\IPhi(U;X)-\IPhi(U;XX)
    \\&=\max_U\IPhi(U;X)
    \\&=\IPhi(X;X)
\end{align*}
where we used Lemma \ref{lmm5ni}.

Since $\wIPhi(X;X)\geq \IPhi(X;X)$ (part i of this lemma), we get $\wIPhi(X;X)=\IPhi(X;X)$.
\begin{align*}
    \wIPhi(X;XY)=\max_{U}\IPhi(U;XY)+\IPhi(U;X)-\IPhi(U;XXY) &=\max_{U}\IPhi(U;X)=\IPhi(X;X).
\end{align*}
Thus $\wIPhi(X;XY)=\wIPhi(X;X)=\IPhi(X;X)$.

Proof of (iv): Take independent $X$ and $Y$. By the chain rule, we have
\begin{align*}
    \wIPhi(X;Y) 
    &= \max_{U} \IPhi(U;X)+\IPhi(U;Y)-\IPhi(U;XY) \\
    &= \max_{U} \IPhi(U;Y)-\IPhi(U;Y|X). 
\end{align*}
By the part (vi) of Lemma \ref{lmm5ni},  $X\independent Y$ implies that $\IPhi(U;Y)-\IPhi(U;Y|X)\leq 0$ for every $U$. This completes the proof.

(v) Suppose $X\to Z\to Y$ is a Markov chain. By the chain rule, we have
\begin{align*}
    \wIPhi(X;Y|Z) 
    &= \max_{U} \IPhi(U;X|Z)+\IPhi(U;Y|Z)-\IPhi(U;XY|Z) \\
    &= \max_{U} \IPhi(U;Y|Z)-\IPhi(U;Y|XZ). 
\end{align*}
By the part (vi) of Lemma \ref{lmm5ni},  the Markov condition implies that $\IPhi(U;Y|Z)-\IPhi(U;Y|XZ)\leq 0$ for every $U$. This completes the proof.

\section{Proof of the properties of the matrix $\Phi$-ribbon} \label{app:proof-matrix-phi}

\begin{proof}[Proof of Lemma \ref{lem:matrix-Phi-ineq}]
    (i) By Definition \ref{def:class-FM} (ii), the mapping
    \begin{align*}
        \mt F(X) \mapsto \sum_{x}p(x)\Phi(\mt F(x)) -\Phi\left(\sum_{x}p(x) \mt F(x)\right)
    \end{align*}
    is convex. Then by Jensen's inequality, for any $q_Y$ and $\mt F(X,Y)$, we have
    \begin{align*}
        &\sum_{y}q(y)
        \Bigg(
        \sum_{x}p(x)\Phi(\mt F(x,y)) -\Phi\left(\sum_{x}p(x) \mt F(x,y)\right)
        \Bigg)\\
        \geq&
        \sum_{x}p(x)\Phi(\sum_y q(y) \mt F(x,y)) -\Phi\left(\sum_{x,y}p(x)q(y) \mt F(x,y)\right),
    \end{align*}
    which is equivalent to
    \begin{align*}
        \HPhi(\mt F|X)\geq \HPhi(\E[\mt F|Y]).
    \end{align*}

    (ii) By (i), fixing $Z=z$, we have the inequalities for the function $\mt G_{XY}^{(z)}(x, y) = \mt F(x,y,z)$. Then taking average over $z$, we obtain the result. 
\end{proof}

\begin{proof}[Proof of Proposition \ref{prop:ribbon-indep}]
    It suffices to prove
    \begin{align*}
        \HPhi(\mt F) \geq \HPhi(\E[\mt F|X])+\HPhi(\E[\mt F|Y]).
    \end{align*}
    By the chain rule \eqref{matrix-Phi:chain-rule}, it is equivalent to showing that 
    \begin{align*}
        \HPhi(\mt F|X)\geq \HPhi(\E[\mt F|Y]),
    \end{align*}
    which is the conclusion of Lemma \ref{lem:matrix-Phi-ineq} \ref{matrix-Phi-ineq-1} since $X$ is independent of $Y$. 
\end{proof}

\begin{proof}[Proof of Theorem \ref{thm:tensor-dpi}]
    By letting $\mt F$ to be an arbitrary mapping of $(X_1,Y_1)$ or $(X_2,Y_2)$, we have
    \begin{align*}
        \fR_{\Phi}^M(X_1X_2;Y_1Y_2) \subseteq \fR_{\Phi}^M(X_1;Y_1)\cap \fR_{\Phi}^M(X_2;Y_2).
    \end{align*}
    To prove the other direction, letting $(\lambda_1,\lambda_2)\in \fR_{\Phi}^M(X_1;Y_1)\cap \fR_{\Phi}^M(X_2;Y_2)$, we then need to show that 
    \begin{align} \label{eq:tensor-direction}
        \HPhi(\mt F)\geq \lambda_1 \HPhi(\E[\mt F|X_1X_2]) + \lambda_2 \HPhi(\E[\mt F|Y_1Y_2]).
    \end{align}
    Since $\E[\mt F|X_1 Y_1]$ is a positive semidefinite mapping of $(X_1,Y_1)$, we have
    \begin{align*}
         \HPhi(\E[\mt F|X_1 Y_1])\geq \lambda_1 \HPhi(\E[\mt F|X_1]) + \lambda_2 \HPhi(\E[\mt F|Y_1]).
    \end{align*}
    Moreover, by fixing $X_1=x$, $Y_1=y$, $\mt F$ is a mapping of $(X_2,Y_2)$. Then taking the expectation, we have
    \begin{align*}
        \HPhi(\mt F|X_1 Y_1)\geq \lambda_1 \HPhi(\E[\mt F|X_1 X_2 Y_1]|X_1 Y_1) +\lambda_2 \HPhi(\E[\mt F|X_1 Y_1 Y_2]|X_1 Y_1).
    \end{align*}
    Note that we have the Markov chains $X_2\rightarrow X_1\rightarrow Y_1$ and $X_1\rightarrow Y_1\rightarrow Y_2$, by Lemma \ref{lem:matrix-Phi-ineq} \ref{matrix-Phi-ineq-2}, we have
    \begin{align*}
        \HPhi(\E[\mt F|X_1 X_2 Y_1]|X_1 Y_1)\geq \HPhi(\E[\mt F|X_1 X_2]),
    \end{align*}
    and
    \begin{align*}
        \HPhi(\E[\mt F|X_1 Y_1 Y_2]|X_1 Y_1)\geq \HPhi(\E[\mt F|Y_1Y_2]),
    \end{align*}
    which proves \eqref{eq:tensor-direction}.

    (ii) Suppose $(\lambda_1,\lambda_2)\in \fR_{\Phi}^M(X_1,Y_1)$. Let $\mt F=\mt F(X_2,Y_2)$. We then have
    \begin{align} \label{dpi-matrix-1}
        \HPhi(\E[\mt F|X_1Y_1])\geq \lambda_1 \HPhi(\E[\mt F|X_1]) +\lambda_2 \HPhi(\E[\mt F|Y_1]).
    \end{align}
    By the assumption, conditioned on $(X_1,Y_1)$, $X_2$ is independent of $Y_2$. By Proposition \ref{prop:ribbon-indep}, we have
    \begin{align}\label{dpi-matrix-2}
        \HPhi(\mt F|X_1Y_1)&\geq \HPhi(\E[\mt F|X_1X_2Y_1]|X_1Y_1) +\HPhi(\E[\mt F|X_1Y_1Y_2]|X_1Y_1)\nonumber\\
        &\geq \lambda_1 \HPhi(\E[\mt F|X_1X_2Y_1]|X_1Y_1) +\lambda_2 \HPhi(\E[\mt F|X_1Y_1Y_2]|X_1Y_1).
    \end{align}
    Summing up \eqref{dpi-matrix-1} and \eqref{dpi-matrix-2}, we have
    \begin{align*}
        \HPhi(\mt F) \geq \lambda_1 \left(\HPhi(\E[\mt F|X_1])+\HPhi(\E[\mt F|X_1X_2Y_1]|X_1Y_1)\right)
        +\lambda_2 \left(\HPhi(\E[\mt F|Y_1])+\HPhi(\E[\mt F|X_1Y_1Y_2]|X_1Y_1)\right).
    \end{align*}
    It remains to prove
    \begin{align*}
        \HPhi(\E[\mt F|X_1])+\HPhi(\E[\mt F|X_1X_2Y_1]|X_1Y_1)\geq \HPhi(\E[\mt F|X_2]).
    \end{align*}
    Since we have the Markov chain $X_2\rightarrow X_1\rightarrow Y_1\rightarrow Y_2$, by Lemma \ref{lem:matrix-Phi-ineq} \ref{matrix-Phi-ineq-2}, we have
    \begin{align*}
        \HPhi(\E[\mt F|X_1X_2Y_1]|X_1Y_1)\geq \HPhi(\E[\mt F|X_1X_2]).
    \end{align*}
    Since conditioning reduces the matrix $\Phi$-entropy (see \eqref{matrix-Phi:cond-reduce}), we have
    \begin{align*}
        \HPhi(\E[\mt F|X_1X_2Y_1]|X_1Y_1) \geq\HPhi(\E[\mt F|X_1X_2]) \geq \HPhi(\E[\mt F|X_2]).
    \end{align*}
    Similarly, we have
    \begin{align*}
        \HPhi(\E[\mt F|X_1Y_1Y_2]|X_1Y_1)\geq \HPhi(\E[\mt F|Y_2]).
    \end{align*}
    Thus, $(\lambda_1,\lambda_2)\in \fR_{\Phi}^M(X_2;Y_2)$, which finishes the proof of data processing property for the matrix $\Phi$-ribbon. 
\end{proof}

\section{Proof of Proposition \ref{prop:M-Phi-SDPI}} \label{app:proof-SDPI}

It suffices to show that 
\begin{align*}
    \lambda^2 \HPhi(\mt F) \geq \HPhi(\E[\mt F|Y]).
\end{align*}
Suppose $\mt F$ takes values from $\mt S_0$ and $\mt S_1$. Let $\mt M=\E[\mt F]$. Then we write $\mt S_0=\mt M-\mt Z$ and $\mt S_1=\mt M+\mt Z$. Thus we will show
\begin{align*}
    \lambda^2 \tr(\Phi(\mt M+\mt Z)+\Phi(\mt M-\mt Z)-2\Phi(\mt M))\geq \tr(\Phi(\mt M+\lambda \mt Z)+\Phi(\mt M-\lambda \mt Z)-2\Phi(\mt M)).
\end{align*}
Let 
\begin{align*}
    \psi(t)= \tr(\Phi(\mt M+\sqrt{t}\mt Z)+\Phi(\mt M-\sqrt{t}\mt Z)-2\Phi(\mt M)).
\end{align*}
It is equivalent to showing that
\begin{align*}
    \lambda^2 \psi(1)\geq \psi(\lambda^2).
\end{align*}
Since $\psi(0)=0$, it suffices to show that $\psi(t)$ is convex in $t$. By the chain rule of Frech\'et derivative, we compute
\begin{align*}
    \psi'(t)=&\tr(\D\Phi(\mt M+\sqrt{t}\mt Z)(\frac{1}{2\sqrt{t}}\mt Z)+\D\Phi(\mt M-\sqrt{t}\mt Z)(-\frac{1}{2\sqrt{t}}\mt Z)
    ). 
\end{align*}
And then
\begin{align*}
    \psi''(t)=&\tr(\D^2\Phi(\mt M+\sqrt{t}\mt Z)(\frac{1}{2\sqrt{t}}\mt Z,\frac{1}{2\sqrt{t}}\mt Z)+\D\Phi(M+\sqrt{t}Z)(-\frac{1}{4t^{3/2}}Z)\\
    &+\D^2\Phi(M-\sqrt{t}Z)(-\frac{1}{2\sqrt{t}}\mt Z,-\frac{1}{2\sqrt{t}}\mt Z)+\D\Phi(M-\sqrt{t}Z)(\frac{1}{4t^{3/2}}Z)).
\end{align*}
Let $s=\sqrt{t}$. It is equivalent to show that
\begin{align*}
    \tr(\D^2\Phi(\mt M+s\mt Z)(\mt Z,\mt Z)+\D^2\Phi(\mt M-s\mt Z)(-\mt Z,-\mt Z))\geq
    \frac{1}{s}\tr(\D\Phi(\mt M+s\mt Z)(\mt Z)-\D\Phi(\mt M-s\mt Z)(\mt Z)).
\end{align*}
Let 
\begin{align*}
    \xi(s)=\tr(\D\Phi(\mt M+s\mt Z)(\mt Z)-\D\Phi(\mt M-s\mt Z)(\mt Z)).
\end{align*}
Thus it suffices to show that $\xi(s)$ is convex for $s\in [0,1]$. 
\begin{align*}
    \xi'(s)=\tr(\D^2\Phi(\mt M+s\mt Z)(\mt Z,\mt Z)+\D^2\Phi(\mt M-s\mt Z)(-\mt Z,-\mt Z)).
\end{align*}
\begin{align*}
    \xi''(s)=\tr(D^3\Phi(\mt M+s\mt Z)(\mt Z,\mt Z,\mt Z)-D^3\Phi(\mt M-s\mt Z)(\mt Z,\mt Z,\mt Z)).
\end{align*}
First we show that $\xi'(s)$ is convex in $s\in [0,1]$, i.e., 
\begin{align*}
    &\lambda \left(\tr(\D^2\Phi(\mt M+s_1\mt Z)(\mt Z,\mt Z)+\D^2\Phi(\mt M-s_1\mt Z)(-\mt Z,-\mt Z))\right)\\
    +&(1-\lambda) \left(\tr(\D^2\Phi(\mt M+s_2\mt Z)(\mt Z,\mt Z)+\D^2\Phi(\mt M-s_2\mt Z)(-\mt Z,-\mt Z))\right)\\
    \geq & \tr(\D^2\Phi(\mt M+(\lambda s_1 +(1-\lambda)s_2)\mt Z)(\mt Z,\mt Z)+\D^2\Phi(\mt M-(\lambda s_1 +(1-\lambda)s_2)\mt Z)(-\mt Z,-\mt Z))
\end{align*}
for $0\leq \lambda \leq 1$ and $s_1,s_2\in [0,1]$. But this is implied by the convexity of 
\begin{align*}
    (A,B) \mapsto \tr(\D^2\Phi(A)(B,B)).
\end{align*}
Then, $\xi''(s)$ is increasing in $s$ and thus
\begin{align*}
    \xi''(s)\geq \xi''(0)=0
\end{align*}
for $s\in [0,1]$. Thus, $\xi(s)$ is convex. It concludes the proof.

\end{document}